\documentclass[12pt]{amsart}

\usepackage{amsmath} 
\usepackage{amsfonts}
\usepackage{amssymb}
\usepackage{amsthm}
\usepackage{latexsym}
\usepackage{color}
\usepackage{enumerate}
\usepackage{mathrsfs}
\usepackage{setspace}
\usepackage{tikz}

\usepackage{amsaddr}

\newtheorem{lemma}{Lemma}[section]
\newtheorem{thm}{Theorem}[section]

\newtheorem{remark}{Remark}[section]
\numberwithin{equation}{section}

\newcommand{\veps}{\varepsilon}

\def\R{\mathbb{R}}

\def\R{\mathbb{R}}

\def\S{\Sigma}
\def\({\left(}
\def\){\right)}

\def\={\stackrel{(n=2)}{=}}
\def\p{\partial}

\newcommand{\be}{\begin{equation}}
\newcommand{\ee}{\end{equation}}
\newcommand{\bee}{\begin{equation*}}
\newcommand{\eee}{\end{equation*}}

\newcommand{\m}{\mathfrak{m}}

\allowdisplaybreaks

\begin{document}
	
	\title[Penrose inequality with charged matter]{On the charged Riemannian Penrose inequality with charged matter}
	
	\author[McCormick]{Stephen McCormick}
	\address{Matematiska institutionen\\ Uppsala universitet\\ 751 06 Uppsala\\ Sweden}
	\email{stephen.mccormick@math.uu.se}
	
	\maketitle

\begin{abstract}
	Throughout the literature on the charged Riemannian Penrose inequality, it is generally assumed that there is no charged matter present; that is, the electric field is divergence-free. The aim of this article is to clarify when the charged Riemannian Penrose inequality holds in the presence of charged matter, and when it does not.
	
	First we revisit Jang's proof of the charged Riemannian Penrose inequality to show that under suitable conditions on the charged matter, this argument still carries though. In particular, a charged Riemannian Penrose inequality is obtained from this argument when charged matter is present provided that the charge density does not change sign. Moreover, we show that such hypotheses on the sign of the charge are in fact required by constructing counterexamples to the charged Riemannian Penrose inequality when these conditions are violated. We conclude by noting that one of these counterexamples contradicts a version of the charged Penrose inequality existing in the literature, and explain how this existing result can be repaired.
\end{abstract}


\section{Introduction}
The Penrose inequality is a conjectured inequality bounding the total (ADM) mass of an asymptotically flat spacetime in terms of the area of its outermost horizon. While the general conjecture remains open, it has been solved in the Riemannian case -- that is, when the spacetime admits a time-symmetric initial data slice -- independently by Huisken and Ilmanen \cite{H-I01} and Bray \cite{Bray01}. The former case used a weak formulation of inverse mean curvature flow, which has proven to be an indispensable tool in geometric analysis and mathematical general relativity, however this technique limits it to only consider the area of one connected component of the horizon and 3 spatial dimensions. The latter case, using a conformal flow of the whole 3-manifold allows for a disconnected horizon and the techniques were later extended up to 7 spatial dimensions by Bray and Lee \cite{Bray-Lee}. For the precise definitions of asymptotic flatness, ADM mass, and other relevant definitions used throughout, the reader is referred to Section \ref{SDefs}. There is a wealth of literature on the Penrose inequality so we only give a brief discussion of it here, however for a good background on the topic the reader is directed to the review articles \cite{MarsPenroseReview} by Mars and \cite{BC} by Bray and Chru\'sciel. For a more recent review on geometric inequalities in general relativity more generally, see the review article \cite{DGC-Review} by Dain and Gabach-Clement.

Penrose's original argument \cite{Penrose} behind this inequality is roughly the following. Given some initial data with ADM mass $\m_o$ and horizon area $A_o$, the total mass cannot increase (although it may decrease if gravitational radiation escapes to infinity) and by Hawking's area theorem, the horizon area, playing the role of entropy, cannot decrease. Therefore if we expect that the long-time evolution tends towards a Kerr solution, whose mass $\m$ and horizon area $A$ satisfy
\be \label{eq-RPI} 
	\m\geq \(\frac{A}{16\pi}\)^{\frac12},
\ee 
then this inequality should also hold for $\m_o$ and $A_o$.

In the case where gravity is coupled to electromagnetism, one also expects a version of this inequality accounting for charge, given by
\be \label{eq-chargedpenroseareaform}
	\m+\sqrt{\m^2-Q^2}\geq \(\frac{A}{4\pi}\)^{\frac12}
\ee 
where $Q$ denotes electric charge and $\m^2\geq Q^2$ assuming the charged version of the positive mass theorem holds \cite{CRT,Gibbons-Hull,GHHP}. In this case, Penrose's heuristic argument roughly becomes the following (see \cite{Jang79}). If the initial data ultimately tends towards a Kerr-Newman solution, which satisfies \eqref{eq-chargedpenroseareaform}, then provided that the charge is conserved, this inequality should hold. Indeed, it has been proven that \eqref{eq-chargedpenroseareaform} holds assuming the charge is entirely shielded by a horizon (minimal surface), that is, there is no charged matter outside of the horizon. From the mathematical point of view, this assumption is that the electric field is divergence-free. This was first proven by Jang \cite{Jang79} assuming what amounts to the existence of a smooth solution to inverse mean curvature flow. This assumption was later rendered superfluous by the development of weak inverse mean curvature flow by Huisken and Ilmanen \cite{H-I01}. This inequality has more recently been proven in the case of multiple black holes by Khuri, Weinstein and Yamada \cite{KWY-2017} under the assumption that a certain charge--area inequality is satisfied, and counterexamples are known in the case of multiple black holes that do not satisfy the charge--area inequality \cite{WY-2005}. In the case of a single (connected) horizon, Jang in fact proved something slightly stronger; he proved 
\be \label{eq-chargedpenrose}
\m\geq \(\frac{A}{16\pi}\)^{\frac12}\(1+\frac{4\pi Q^2}{ A}\),
\ee 
which also implies a lower bound on the area of the horizon. Throughout this article we will generally mean \eqref{eq-chargedpenrose} when we speak of the charged Riemannian Penrose inequality. Despite this success, very little has been said about the case where charged matter is present. In this note, we aim to clarify to what extent the charged Riemannian Penrose inequality holds in the presence of charged matter; that is, when the electric field is not divergence-free.

Note that in \eqref{eq-chargedpenrose} and \eqref{eq-chargedpenroseareaform}, there is some intentional ambiguity in regard to where $Q$ is measured. If there is no charged matter other than what may be shielded behind the horizon, then the charge of the horizon is the same as the charge at infinity, which simply follows from the divergence theorem. However, this is certainly not true in general if charged matter is present. Assume first that $Q>0$ refers to the charge measured on the horizon and there is charged matter exterior to the horizon. Following Penrose's original argument, if some negatively charged matter were to fall through the horizon then the charge of the horizon would decrease, in which case there is no reason to expect that just because \eqref{eq-chargedpenroseareaform} holds for Reissner--Nordstr\"om that it should hold generally. In fact, a counterexample is given in Section \ref{SChargedPenrose} by Theorem \ref{thm-counterexampleBH}, demonstrating that \eqref{eq-chargedpenroseareaform} fails to hold if the charge of the horizon and that of the matter have opposite sign. However, if we only permit positive charges outside of the horizon then Penrose's argument still holds. That is, we only require that the total charge cannot decrease. Indeed, \eqref{eq-chargedpenrose} follows from Jang's argument in this case and is given by Theorem \ref{thm-BHQ}.

On the other hand, if we take $Q=Q_\infty>0$ to be the total charge measured at infinity in \eqref{eq-chargedpenrose}, then charged matter falling into the horizon can not change the total charge but some charged matter might instead escape to infinity. In this case, following Penrose's argument, one may hope that \eqref{eq-chargedpenrose} still holds provided that there is no positively charged matter outside of the horizon. In this case any charge that escapes to infinity could at worst increase the total charge. Again, it is shown that the inequality holds in this case (Theorem \ref{thm-inftyQ}). Additionally, without restricting the sign of the charge density, a counterexample is given (Theorem \ref{thm-counterexampleinfty}). Of particular note, Theorem \ref{thm-counterexampleinfty} suggests that a correction should be made to a version of the charged Penrose inequality obtained by Khuri, Weinstein and Yamada \cite{KWY-exts-2015}, which we discuss in Section \ref{SSPMT}.

In Section \ref{SDefs}, we recall some standard background and definitions. The crux of this note is then contained in Section \ref{SChargedPenrose}, where we show the two versions of the charged Penrose inequality with charged matter and then construct our counterexamples mentioned above. Both versions of the charged Penrose inequalities follow directly from the argument of Jang \cite{Jang79}, with only minor modifications to track the charged matter, combined with the subsequent development of weak inverse mean curvature flow by Huisken and Ilmanen \cite{H-I01}. The counterexamples that we construct are spherically symmetric and are simply obtained by smoothly gluing Schwarzschild and Reissner--Nordstr\"om manifolds together. We conclude by discussing the relationship between these counterexamples and \cite{KWY-exts-2015}.

 For the sake of exposition, we assume the magnetic field vanishes throughout and discussion of magnetic fields is relegated to the appendix.

\section{Background and definitions} \label{SDefs}
In this section, we briefly recall some standard definitions. We say a Riemannian $3$-manifold $(M,g)$ is \emph{asymptotically flat} if there exists a compact set $K$ such that $M\setminus K$ is diffeomorphic to $\R^3$ minus a closed ball, and such that $g$ is asymptotic to a flat metric on $M\setminus K$ in the following sense: let $\delta$ denote the pullback of the Euclidean metric to $M\setminus K$ and in the Cartesian coordinates $x^i$ given by this diffeomorphism we ask that $g$ satisfies
\be 
	g-\delta=O(|x|^{-1})\qquad \p g=O(|x|^{-2}) \qquad \p^2 g=O(|x|^{-3})
\ee 
as $|x|\to\infty$. An asymptotically flat manifold $(M,g)$ can be seen as time-symmetric initial data in the context of general relativity, representing an isolated gravitating system. This system can be said to have a total mass given by the ADM mass \cite{ADM}, which is a geometric property of the manifold computed by its asymptotics. It is a well-defined geometric quantity \cite{Bartnik-86,Chrusciel-86}, however it is most conveniently expressed in terms of the aforementioned Cartesian coordinates near infinity as
\be \label{eq-ADM}
	\m_{ADM}(M,g)=\frac{1}{16\pi} \lim\limits_{R\to\infty}\int_{S_R} \(\p_i g_{ij}-\p_j g_{ii}\)dS^i,
\ee
where repeated indices are summed over and $S_R$ is a large coordinates sphere of radius $R$. In the context of time-symmetric initial data a \emph{horizon} is a closed minimal surface, and if a stable closed minimal surface exists then this implies the existence of an event horizon in the evolution. We say that a horizon, or indeed any closed surface $\S$ is \emph{outer minimising} if there are no surfaces enclosing it with less area. The Riemannian Penrose inequality \eqref{eq-RPI} gives a lower bound on the mass in terms of an outer minimising horizon, provided that $(M,g)$ has nonnegative scalar curvature.

Note that the nonnegativity of scalar curvature amounts to an energy condition on the matter source. Specifically, initial data for general relativity must satisfy certain constraint equations and in the time-symmetric case, these constraints are simply
\be 
	R(g)=16\pi \mu,
\ee 
where $\mu$ is the local energy density of the matter fields, which we assume to be nonnegative (imposing the dominant energy condition).

Generally, one also asks that the matter fields satisfy equations of their own, in which case we generally have a coupled system of equations. Of particular importance to the present article is the case where some of the matter corresponds to electromagnetic fields. In this case, we may speak of time-symmetric initial data for the Einstein-Maxwell equations $(M,g,E,B)$ where $E$ and $B$ are vector fields on $M$ corresponding to the electric and magnetic fields respectively. The constraint equations are then
\begin{align}\begin{split}
	R(g)-2\(|E|^2+|B|^2\)&=\mu\\
	\nabla\cdot E&=4\pi\rho\\
	\nabla\cdot B&=0,\end{split}
\end{align}
where $\mu$ is the energy density of the other matter and $\rho$ is the electric charge density of the matter. Note, it is customary to view $B$ as arising from a vector potential and the equation $\nabla\cdot B=0$ is a consequence of this, rather than a genuine constraint. Throughout this article we will set $B=0$ for the sake of exposition, however the results remain true in the presence of magnetic fields (see Appendix \ref{AMagnetic} for details). We will say that a triple $(M,g,E)$ is a \emph{charged asymptotically flat manifold}, provided that $E$ and $\p E$ decay as $O(|x|^{-2})$ and $O(|x|^{-3})$ respectively. We say the \emph{charged dominant energy condition} is satisfied if $R(g)\geq 2|E|^2$. It should be remarked that there is another (stronger) version of the charged dominant energy condition appearing in the literature, which is discussed in more detail in Section \ref{SSPMT}.

For a closed surface $\S$ in $M$, we define the charge enclosed by $\S$ to be given by the flux integral
\be 
Q(\S)=\frac{1}{4\pi}\int_\S E\cdot n\,dS.
\ee 

We also define the total charge of the manifold to be
\be 
Q_\infty=\frac{1}{4\pi}\lim\limits_{R\to\infty}\int_{S_R} E\cdot n\,dS.
\ee 

The charged Riemannian Penrose inequality for a single horizon then follows from an argument of Jang \cite{Jang79} combined with the later development of weak inverse mean curvature flow by Huisken and Ilmanen \cite{H-I01} and can be stated as the following.
\begin{thm}
	Let $(M,g,E)$ be a charged asymptotically flat manifold satisfying the charged dominant energy condition with no charged matter (that is, $\rho\equiv0$). Let $\S$ be a connected outermost minimal surface and assume there exist no other minimal surfaces outside of $\S$, then we have
	\be \label{eq-chargedpenrose1}
	\m_{ADM}\geq \(\frac{|\S|}{16\pi}\)^{\frac12}\(1+\frac{4\pi Q^2}{ |\S|}\),
	\ee 
	where $|\S|$ denotes the area of $\S$ and $Q=Q_\infty=Q(\S)$.
\end{thm}
Assuming that the charge and area of $\S$ satisfy $4\pi Q^2\leq |\S|$ then this was also shown to hold in the case the $\S$ is not connected by Khuri, Weinstein and Yamada \cite{KWY-2017}. In the case that the inequality $4\pi Q^2\leq |\S|$ does not hold and $\S$ is not connected, there in fact exist counterexamples, as shown by Weinstein and Yamada \cite{WY-2005}.

\section{The Riemannian Penrose inequality with charged matter} \label{SChargedPenrose}

As mentioned above, the charged Riemannian Penrose inequality has been known for some time, however it is usually stated in the case $\rho\equiv0$; that is, without any charged matter. In \cite{Jang79} Jang's argument was mostly heuristic, since it effectively relied on the existence of a smooth solution to inverse mean curvature flow. This may be why Jang did not consider the inclusion of charged matter, or perhaps because this heuristic argument was given in the context of the cosmic censorship conjecture and Jang went on to discuss issues with the formulation of the conjecture if matter is present. Regardless of the reason for not considering charged matter at the time, since the development of weak inverse mean curvature flow, Jang's argument is viewed as a genuine proof throughout the literature, and the hypothesis excluding charged matter has persevered. Here we recall the details of Jang's proof, keeping track of the charged matter for the sake of clarifying this point.

\subsection{Jang's argument with charged matter}

We consider the notion of a charged Hawking mass of a topological sphere $\S$ (cf. \cite{ACC19,DK-2015}), given by
\be 
\m_H^{CH}(\S)=\( \frac{|\S|}{16\pi}\)^{\frac12}\( 1+\frac{4\pi Q(\S)^2}{|\S|}-\frac{1}{16\pi}\int_\S H^2\, dS \),
\ee 
where $Q(\S)$ is the total electric charge contained within $\S$. and $H$ is the mean curvature of $\S$. With this in mind, we show the following.
\begin{thm}\label{thm-BHQ}
	Let $(M,g,E)$ be a charged asymptotically flat 3-manifold satisfying the charged dominant energy condition and let $\S$ be an area outer minimising sphere and assume there are no closed minimal surfaces in $M$ except possibly $\S$. Assume that exterior to $\S$ it holds that $Q(\S)\nabla\cdot E\geq0$; that is, the charge density of the matter fields either vanishes or is everywhere the same sign as the charge of $\S$. 
	
	 Then
	\be \label{eq-Hawkingineq}
		\m_H^{CH}(\S)\leq\m_{ADM},
	\ee 
 and in particular, if $\S$ is an outermost horizon then
	
	\be\label{eq-chargedpenrose2} 
		\m_{ADM}\geq \(\frac{|\S|}{16\pi}\)^{\frac12}\(1+\frac{4\pi Q^2}{ |\S|}\),
	\ee 
	where $Q=Q(\S)$.
	
	Furthermore, equality holds if and only if $\S$ is a round sphere in a Reissner--Nordstr\"om manifold.
\end{thm}
\begin{proof}
We first see that a (modified) charged Hawking mass is monotone under the smooth inverse mean curvature flow, provided that $Q(\S)\nabla\cdot E\geq0$. In particular we show that for a smooth inverse mean curvature flow, starting on a surface $\S_o$, the quantity
\be 
\m_o^{CH}(\S)=\( \frac{|\S|}{16\pi}\)^{\frac12}\( 1+\frac{4\pi Q(\S_o)^2}{|\S|}-\frac{1}{16\pi}\int_\S H^2\, dS \),
\ee 
is monotone provided that the charge density in the region flowed through has the same sign as $Q(\S_o)$. Note that the charge in the above expression is simply the initial charge. Recall that a solution to (smooth) inverse mean curvature flow is a family of surfaces $\{ \S_t\}=X(\S,t)$ satisfying
\bee 
	\frac{\p X}{\p t}=\frac{1}{H} n,
\eee 
where $H>0$ and $n$ is the outward unit normal to $\S_t$.

Using the shorthand $Q=Q(\S_o)$, a well-known computation gives (See \cite{H-I01}, page 395--396 for details)
\be 
	\frac{d}{dt}\( \int_{\S_t} H^2\,dS_t\)\leq \frac12\(16\pi - \int_{\S_t} H^2\,dS_t \)-\int_{\S_t} R\,dS_t,
\ee 
where $R=R(g)$ is the scalar curvature of $g$. Combining this with the charged dominant energy condition gives

\bee
	\frac{d}{dt}\( \int_{\S_t} H^2\,dS_t- \frac{64\pi^2Q^2}{|\S_t|}\)\leq \frac12\(16\pi - \int_{\S_t} H^2\,dS_t\)-2\int_{\S_t}|E|^2\,dS_t+\frac{64\pi^2Q^2}{|\S_t|}.
\eee 
Now let $\Omega_t$ be the region enclosed between $\S_o$ and $\S_t$, and following Jang's original argument \cite{Jang79} we have
\begin{align*}
	\frac{64\pi^2Q^2}{|\S_t|}-2\int_{\S_t}|E|^2\,dS&\leq \frac{64\pi^2Q^2}{|\S_t|}-\frac{2}{|\S_t|}\(\int_{\S_t}|E|\,dS_t\)^2\\
	&\leq\frac{2}{|\S_t|}\(32\pi^2Q^2-\(\int_{\S_t}E\cdot n\,dS_t\)^2\)\\
	&\leq\frac{2}{|\S_t|}\(32\pi^2Q^2-\(4\pi Q+\int_{\Omega_t}\nabla\cdot E\,dV\)^2\)\\
	&\leq\frac{2}{|\S_t|}\(16\pi^2Q^2-8\pi Q\int_{\Omega_t}\nabla\cdot E\,dV-\(\int_{\Omega_t}\nabla\cdot E\,dV\)^2\).
\end{align*}
So, provided that $\int_{\Omega_t}\nabla\cdot E\,dV$ and $Q$ have the same sign, we have
\be \label{eq-evo}
\frac{d}{dt}\(1- \frac{1}{16\pi}\int_{\S_t} H^2\,dS_t+ \frac{4\pi Q^2}{|\S_t|}\)\geq -\frac12\(1- \frac{1}{16\pi}\int_{\S_t} H^2\,dS_t+ \frac{4\pi Q^2}{|\S_t|}\).
\ee
Since $\frac{d}{dt}|\S_t|^{1/2}=\frac12|\S_t|^{1/2}$, this in turn demonstrates that $\m_o^{CH}(\S_t)$ is monotonically non-decreasing, and strictly increasing somewhere if $\nabla\cdot E$ is not identically zero.

It is clear that if the smooth flow exists for all time then this quantity has the same limit as the Hawking mass, namely the ADM mass and therefore flowing out from the horizon would establish the inequality. However, this is not the case and one must employ the weak inverse mean curvature flow of Huisken and Ilmanen \cite{H-I01}. During the weak flow, the surface $\S_t$ may at some point `jump' to a minimising hull. Specifically, at certain times through the flow, the surface may fail to be area outer minimising, which means there is a time where some surface containing $\S_t$ has the same area as $\S_t$, and in particular the new surface is outer minimising. At such a time, the flow simply jumps to this new surface and then continues. It is well-known \cite{H-I01} that, in the absence of obstacles (that is, other boundary components or minimal surfaces in $M$) the area of the new surface is equal to the area of the surface prior to the jump. In fact, the area of $\S_t$ along the flow satisfies $|\S_t|=|\S_o|\exp(t)$, even for the weak formulation of the flow. Since $Q$ is constant and $\S_t$ is continuous along the flow, the monotonicity extends to the weak formulation of inverse mean curvature flow. In order to see this more clearly, recall the Geroch monotonicity formula for the weak inverse mean curvature flow, given by equation (5.24) of \cite{H-I01}:
\be 
	\m_H(\S_s)-\m_H(\S_r)\geq \frac{1}{(16\pi)^{3/2}}\int_r^s|\S_t|^{1/2}\int_{\S_t}(\Xi+R)\,d\mu_t dt
\ee 
where $\m_H(\S_t)=\m_o^{CH}(S_t)-\sqrt{\frac{\pi}{|\S_t|}}Q(\S_o)^2$ is the usual (uncharged) Hawking mass and $\Xi$ denotes a collection of nonnegative terms. Note that this formula is valid for any flow times $s>r\geq0$. Following the smooth case, we see that we have
\begin{align*}
\m_o^{CH}(S_s)-\m_o^{CH}(S_r)&\geq\frac{\sqrt{\pi}Q(\S_o)^2}{\sqrt{|\S_s|}}-\frac{\sqrt{\pi}Q(\S_o)^2}{\sqrt{|\S_r|}}+ \frac{2}{(16\pi)^{3/2}}\int_r^s|\S_t|^{1/2}\int_{\S_t}|E|^2d\mu_t dt\\
&\geq \frac{\sqrt{\pi}Q(\S_o)^2}{\sqrt{|\S_s|}}-\frac{\sqrt{\pi}Q(\S_o)^2}{\sqrt{|\S_r|}}+ \frac{2}{(16\pi)^{3/2}}\int_r^s|\S_t|^{-1/2}\(\int_{\S_t}E_nd\mu_t\)^2 dt\\
&\geq \frac{\sqrt{\pi}Q(\S_o)^2}{\sqrt{|\S_s|}}-\frac{\sqrt{\pi}Q(\S_o)^2}{\sqrt{|\S_r|}}+ \frac{32\pi^2Q(\S_o)^2}{(16\pi)^{3/2}}\int_r^s|\S_t|^{-1/2} dt,
\end{align*}
where the estimate for $\int_{\S_t}|E|^2 d\mu_t$ follows exactly as in the smooth case, also making use of the fact that the charges have the same sign. Notice that the final term in the above expression cancels exactly with the first two terms, since $|\S_t|=|\S_o|\exp(t)$. Therefore, monotonicity also holds for the weak flow giving \eqref{eq-Hawkingineq}, which in turn implies \eqref{eq-chargedpenrose2}.

To prove rigidity, we simply must note that as in \cite{H-I01} (see page 422 therein for more details), we must have equality in each step of the monotonicity computation and therefore all of the discarded terms must vanish everywhere. In particular $\nabla\cdot E$ must be identically zero. We can then directly apply the rigidity statement for the case where no charged matter is present, established by Disconzi and Khuri \cite{DK-2015}.
\end{proof}
By reversing the sign of the charge density associated with matter and taking $Q$ to be the total charge rather than the black hole charge, a similar result holds.

\begin{thm} \label{thm-inftyQ}
	Let $(M,g,E)$ be a charged asymptotically flat 3-manifold satisfying the charged dominant energy condition and let $\S$ be an outermost minimal sphere. Assume further that exterior to $\S$ it holds that $Q_\infty\nabla\cdot E\leq0$ and there are no closed minimal surfaces in $M$ except for $\S$.
	
	 Then
	\bee
	\m_{ADM}\geq \(\frac{|\S|}{16\pi}\)^{\frac12}\(1+ \frac{4\pi Q_\infty^2}{ |\S|}\),
	\eee 
	where $Q_\infty$ is the total charge of the manifold. Furthermore, equality holds if and only if $\S$ is a round sphere in a Reissner--Nordstr\"om manifold.
\end{thm}

\begin{proof}
	This follows almost identically to the proof of Theorem \ref{thm-BHQ}. Let $E_t$ be the region exterior to $\S_t$. Following the argument in the proof of Theorem \ref{thm-BHQ}, we see
	\begin{align*}
	\frac{64\pi^2 Q^2_\infty}{{|\S_t|}}-2\int_{\S_t}|E|^2\,dS&\leq\frac{2}{|\S_t|}\(32\pi^2 Q_\infty^2-\(\int_{\S_t}E\cdot n\,dS_t\)^2\)\\
	&\leq\frac{2}{|\S_t|}\(32\pi^2 Q_\infty^2-\(4\pi Q_\infty-\int_{E_t}\nabla\cdot E\cdot n\,dV\)^2\)\\
	&\leq\frac{32\pi^2 Q_\infty^2}{|\S_t|},
	\end{align*}
	where we make use of the fact that $Q$ has the opposite sign to the charge density of the matter. It is easily verified that the computation for the weak flow is the same and therefore for the same reasons as in the preceding theorem, we have monotonicity under the weak flow of the quantity
	\be 
	\m_\infty^{CH}(\S)=\( \frac{|\S|}{16\pi}\)^{\frac12}\( 1+\frac{4\pi Q_\infty^2}{|\S|}-\frac{1}{16\pi}\int_\S H^2\, dS \).
	\ee 
	This gives the desired result, noting that rigidity is identical to Theorem \ref{thm-BHQ}.
\end{proof}
It is worth remarking that there are also inequalities in the literature relating the size of a matter body -- as opposed to a black hole -- to its total mass and charge \cite{Dain,JK,Reiris}. In each of these results, it is also assumed that there exists no charged matter outside of the matter body in question. However, for the same reasons outlined above, one should expect that this can be relaxed this to allow charged matter outside of body in question provided that it satisfies a sign hypothesis similar to the above. Particularly as such results also follow from inverse mean curvature flow and the monotonicity of the Hawking mass \cite{Dain}. However, we omit the details here, as it is outside of the scope of this article.

\subsection{Counterexamples where the charge density has opposite sign}
We now turn to discuss the case where the charge of the matter does not have the correct sign, and indeed provide counterexamples to \eqref{eq-chargedpenrose} if these hypotheses are not satisfied. To motivate the counterexamples, consider placing a large spherical shell of charged matter outside of a Schwarzschild black hole, far from the horizon. It can be shown that the total mass can be kept close to the mass of the original Schwarzschild solution $\widetilde m=\(\frac{A}{16\pi}\)^{\frac12}$, while the total charge can be made large. In particular, we sketch how a counterexample to Theorem \ref{thm-inftyQ} could be possible without a hypothesis on the sign of the charge density.

This first counterexample that we construct is not smooth along a hypersurface, exhibiting an electric charge in the distributional sense along a thin shell -- a sphere $S_R$ of large radius $R$. While such a construction is not directly a counterexample since one would usually assume smooth solutions, it motivates the results to follow. Since the charge density satisfies $\rho=\frac{1}{4\pi}\nabla\cdot E$, it can be seen that the charge density will have a spike in the distributional sense if the radial flux $E_r=E\cdot\p_r$ is discontinuous along $S_R$. In particular, this may be achieved by considering a (Riemannian) Schwarzschild manifold cut off at some large radius $R$ then identifying this spherical boundary with a sphere of radius $R$ in a Reissner--Nordstr\"om manifold set.

So, consider a Schwarzschild manifold with mass ${\widetilde m}$ and note that it can be expressed in coordinates as 
\be 
	g_{\widetilde m}=ds^2+u_{\widetilde m}(s)^2g_*
\ee 
where $g_*$ is the standard round metric of area $4\pi$ and $u_{\widetilde m}$ is the area radius as a function of the coordinate $s$, satisfying
$$ u_{\widetilde m}'(s)=\sqrt{1-\frac{2{\widetilde m}}{u_{\widetilde m}(s)}}.$$
Similarly the Reissner--Nordstr\"om manifold of mass $m$ and electric charge $Q$ can be expressed in coordinates as
\be 
g_{m,Q}=ds^2+v_{m,Q}(s)^2g_*
\ee 
where $v_{m,Q}$ satisfies
$$ v_{m,Q}'(s)=\sqrt{1+\frac{Q^2}{v_{m,Q}^2}-\frac{2m}{v_{m,Q}(s)}}.$$ Note that the level sets $\{u_{\widetilde m}=r\}$ and $\{v_{m,Q}=r\}$ are spheres of radius $r$, so for the sake of exposition we use the coordinate `$r$' like this on both manifolds.

In order to satisfy the dominant energy condition, we must ensure that the scalar curvature does not have any negative spikes in the distributional sense. This puts us in the realm of the positive mass theorem with corners \cite{Miao02} and the correct condition to ensure this is that mean curvatures match. This condition is encoded in conditions (ii) and (iii) in Lemma \ref{gluing-lemma}, below. It is easy to check that in order for the mean curvatures of both metrics to agree on a sphere of area radius $R$, we must have
\be 
\widetilde m = m - \frac{Q^2}{R}.
\ee 
In particular, if we fix the interior Schwarzschild manifold and fix the charge we would like at infinity, we are free to choose $m$ and the radius $R$ at which we perform the gluing. Specifically, for fixed $Q$ and $\widetilde m$ we are able to choose $m$ arbitrarily close to $\widetilde m = \(\frac{A}{16\pi}\)^{\frac12}$. Since $Q$ and $A$ are fixed, this would violate \eqref{eq-chargedpenrose}.

As it stands, this construction does not immediately contradict any inequalities as mentioned above. Nevertheless, we are able to construct a smooth analogue of this example corresponding to a spherically symmetric charge distribution contained in a thin annular region. To do this, we apply the following lemma, which is essentially Lemma 2.1 of \cite{CCMM} (or Lemma 2.2 of \cite{M-S}) combined with Lemma 2.3 of \cite{M-S} to relax condition (i) below to nonnegative scalar curvature rather than strictly positive.

	\begin{lemma} \label{gluing-lemma}
	Let $f_i:[a_i,b_i]\to \R{^{+}}$, where $i=1,2$, be smooth positive functions, and let $g_*$ be the standard metric on $\mathbb{S}^n$. Suppose that
	\begin{itemize}
		\item[(i)]  the metrics $\gamma_{i}:= dt^2+f_i(t)^2g_*$ have nonnegative scalar curvature; 
		\item[(ii)] $f_1(b_1)<f_2(a_2)$; 
		\item[(iii)]    $ 1 > f_1'(b_1) > 0 $ and $ f_1'(b_1)  \geq f_2'(a_2) >  -1  $. 
	\end{itemize}
	Then,  after translating the intervals one can construct a smooth positive function $f:[a_1,b_2]\to\R{^{+}}$ so that:
	\begin{itemize}
		\item [(I)] $f\equiv f_1$ on $[a_1,\frac{a_1+b_1}{2}]$, $f\equiv f_2$ on $[\frac{a_2+b_2}{2},b_2]$, and
		\item[(II)]  $\gamma:= dt^2+f(t)^2g_*$ has nonnegative scalar curvature on $[a_1,b_2] \times \mathbb{S}^n$.
	\end{itemize}
\end{lemma}

This allows us to prove:

\begin{thm} \label{thm-counterexampleinfty}
	For any $A>0, Q>0$ and $m>\(\frac{A}{16\pi}\)^{\frac12}$ there exists a charged asymptotically flat manifold $(M,g,E)$ satisfying the charged dominant energy condition, with mass $m$, asymptotic charge $Q$ and outermost minimal surface with area $A$.
\end{thm}
\begin{proof}
	Let $\widetilde m=\(\frac{A}{16\pi}\)^{\frac12}$ and consider the Schwarzschild manifold of mass $\widetilde m$, as above. Again, excise the region outside of $r=R$ and we will glue this smoothly to a Reissner--Nordstr\"om exterior using Lemma \ref{gluing-lemma}. Note that both the Schwarzschild manifold and the Reissner--Nordstr\"om manifold satisfy condition (i) of the lemma. In order to create a small amount of space for the smooth gluing, we consider the exterior Reissner--Nordstr\"om manifold outside of a radius $r=R':=R+\frac1R$, which amounts to satisfying condition (ii). We would again like the mean curvatures to match, or rather that the mean curvature can at worst drop across the gluing region. That is, we ask 
	\be 
		\frac{\widetilde m}{R} \leq \frac{m}{R'} - \frac{Q^2}{R'^2},
	\ee 
	where $m$ and $Q$ are mass and charge of the Reissner--Nordstr\"om manifold, respectively. Note that in the context of Lemma \ref{gluing-lemma}, this is condition (iii). Since $\widetilde m<m$, this is easily achieved by sufficiently large $R$. Applying the lemma results in an asymptotically flat manifold that is exactly equal to a Schwarzschild manifold of mass $\widetilde m$ near the horizon and exactly equal to a Reissner--Nordstr\"om manifold of mass $m$ and charge $Q$ near infinity.

Of course, without equipping our manifold with an electric field $E$, we cannot speak of the electric charge yet. Naturally though, we would like the electric field to agree with the Reissner--Nordstr\"om electric field $E_o$ associated to our exterior region. Define an annular region $D:=\{ R_2<r<R_2+\frac{1}{R_2} \}$ where $R_2>R'$ and a cut-off function $\chi=\chi(r)$ that is equal to $0$ for $r<R_2$ and equal to $1$ for $r>R_2+\frac{1}{R_2}$, with $\chi'\geq0$. We now fix the electric field to be $E=\chi E_o$, ensuring that the field vanishes everywhere that the metric is not identically Reissner--Nordstr\"om and yields the same asymptotic electric charge $Q$. It is easy to check that the charged dominant energy condition is satisfied everywhere, that is,
$$ R(g)\geq 2|E|^2. $$

This follows from the fact that $R(g)$ is nonnegative everywhere and outside $r=R'$ we have $R(g)=2|E|^2\geq 2|E_o|^2$.

\end{proof}
It is clear that $\nabla\cdot E$ is nonvanishing over the annular region $D$ in the above example, and in fact it resembles a thin annular region of positive charge. One could in principle make this annular region arbitrarily thin, approximating the nonsmooth example given first. It should be noted that the exterior regions considered in the above example are permitted to be the regions exterior to superextremal Reissner--Nordstr\"om solutions.
\begin{remark}
	Note that this construction allows us to construct solutions with arbitrarily large electric charge and with horizon area arbitrarily close to the optimal permitted by the Riemannian Penrose inequality; that is, the uncharged inequality. This should be compared to the work of Mantoulidis and Schoen \cite{M-S}, where they showed that given and metric on the $2$-sphere satisfying a certain stability condition one can construct asymptotically flat manifolds with ADM mass arbitrarily close to the optimal mass given by the Riemannian Penrose inequality. These extensions are exactly isometric to Schwarzschild manifolds outside of a compact set, so this combined with the above construction shows the following: if $g$ is a metric on the $2$-sphere $\mathcal{S}^2$ and the operator $L=-\Delta_g+\frac12 R(g)$ has positive first eigenvalue, then for any $Q>0$ and $\veps>0$ there exists a charged asymptotically flat manifold with boundary isometric to $(\mathcal{S}^2,g)$, total charge $Q$ and mass $\m<\(\frac{|\mathcal{S}^2|_g}{16\pi}\)^{1/2}+\veps$. This may be compared to a recent result of Alaee, Cabrera Pacheco, and Cederbaum \cite{ACC19} which establishes the following: under the same eigenvalue hypothesis on $g$ as used in \cite{M-S}, and for $Q$ not too large relative to $|\mathcal{S}^2|$, one can construct asymptotically flat manifolds with boundary isometric to $(\mathcal{S}^2,g)$, total charge $Q$, vanishing charge density, and mass $\m<\(\frac{|\mathcal{S}^2|_g}{16\pi}\)^{1/2}+\(\frac{\pi}{|\mathcal{S}^2|_g}\)^{1/2}Q^2+\veps$.
\end{remark}

If we view \eqref{eq-chargedpenrose} as referring to the black hole charge as is the case in Theorem \ref{thm-BHQ}, then we again see that the hypothesis on the sign of the charge density is required.

\begin{thm} \label{thm-counterexampleBH}
	For any $A>0$ and $Q>0$ there exists a charged asymptotically flat manifold $(M,g,E)$ satisfying the charged dominant energy condition, having ADM mass ${\widetilde m<\(\frac{A}{16\pi}\)^{1/2}\(1+\frac{4\pi Q^2}{A}\)}$, black hole charge $Q$, and outermost minimal surface whose area is $A$.
\end{thm}
\begin{proof}
	Let $m=\(\frac{A}{16\pi}\)^{1/2}\( 1+ \frac{4\pi Q^2}{A}\)$ and consider the Reissner--Nordstr\"om manifold of mass $m$ and charge $Q$.

	 Again, excise the region outside of $r=R>R_H$ and we will now glue this smoothly to a Schwarzschild exterior using Lemma \ref{gluing-lemma}. Consider the exterior Schwarzschild manifold outside of a radius $r=R':=R+\veps$, for $\veps>0$ to be chosen small. The mean curvature matching condition is effectively the same as the above, but reversed:
	\be 
	\frac{\widetilde m}{R'} \geq \frac{m}{R} - \frac{Q^2}{R^2},
	\ee 
	where $\widetilde m$ is the mass of the Schwarzschild manifold that we are gluing to the exterior. We therefore choose
	\be 
		\widetilde m=R'\( \frac{m}{R} - \frac{Q^2}{R^2} \)=\frac{R'}{R}\(\frac{A}{16\pi}\)^{\frac12}\(1+\frac{4\pi Q^2}{A}\)-\frac{R'Q^2}{R^2}.
	\ee 
	
	Note that the mass $\widetilde m$ of the Schwarzschild exterior satisfies $$\widetilde m<\(\frac{A}{16\pi}\)^{\frac12}\(1+\frac{4\pi Q^2}{A}\),$$ provided that $\veps$ is small.
	
 Again we may apply Lemma \ref{gluing-lemma}, and we obtain an asymptotically flat manifold that is exactly equal to a Reissner--Nordstr\"om manifold of mass $m$ and charge $Q$ near the horizon and exactly equal to a Schwarzschild manifold of mass $\widetilde m$ near infinity.
	
	As above, we interpolate the electric field between the two manifolds somewhere that the new manifold is exactly Reissner--Nordstr\"om. In this case, the resultant electric field satisfies $\nabla\cdot E=0$ everywhere except the transition region where $\nabla\cdot E\leq0$.
	
\end{proof}

\subsection{Other hypotheses for the charged Penrose inequality}\label{SSPMT}
The Riemannian Penrose inequality is generally seen as a strengthened version of the positive mass theorem. When one speaks of the charged positive mass theorem, one usually means to refer to the inequality
\be \label{eq-QPMT}
	\m_{ADM}\geq |Q|.
\ee 
Again, it is quite common to state this result in the case $\nabla\cdot E\equiv 0$ so one should be careful when referring to the literature. It was first shown by Gibbons and Hull \cite{Gibbons-Hull} for black holes with no surrounding charged matter, and then by Gibbons, Hawking, Horowitz and Perry \cite{GHHP} in the case where charged matter is present (see also \cite{CRT}). Naturally, this inequality should also make sense without the presence of black holes so one generally considers $Q=Q_\infty$, the total charge in this inequality. Clearly we can write down initial data with very small mass density and very large charge density, and then one should not expect this inequality to hold. Indeed, we can see this inequality is violated by the counterexample given by Theorem \ref{thm-counterexampleinfty}, for sufficiently large $Q$. However, if we additionally impose that \eqref{eq-QPMT} holds in some sense for the matter field densities, then indeed \eqref{eq-QPMT} holds for the asymptotically defined quantities. Specifically, if one imposes the pointwise condition
\be \label{eq-QDEC2}
	R(g)-2|E|^2\geq 4|\nabla\cdot E|
\ee 
then \eqref{eq-QPMT} holds. Note that \eqref{eq-QDEC2} is a stronger condition than the charged dominant energy condition we use here.

	Recently, Khuri, Weinstein and Yamada gave a version of the charged Riemannian Penrose inequality allowing for charged matter, provided that the charged matter is compactly supported (Theorem 1.3 of \cite{KWY-exts-2015}). The proof therein relies on their earlier work on the charged Riemannian Penrose inequality for multiple black holes \cite{KWY-2017}, without the presence of charged matter. The proof is an adaptation of Bray's conformal flow method \cite{Bray01} to the Einstein--Maxwell case and crucially makes use of the charged positive mass theorem \eqref{eq-QPMT}. However, as mentioned above, \eqref{eq-QPMT} fails to hold under the standard charged dominant energy condition and in fact requires the stronger hypothesis given by \eqref{eq-QDEC2}. Indeed, Theorem \ref{thm-counterexampleinfty} above demonstrates that Theorem 1.3 of \cite{KWY-exts-2015} requires a stronger hypothesis than the standard charged dominant energy condition. However, if one replaces the charged dominant energy condition with \eqref{eq-QDEC2} then the analysis in \cite{KWY-2017} carries over with some minor modifications, and therefore the conclusion of Theorem 1.3 in \cite{KWY-exts-2015} holds. In particular, if one replaces equation (2.1) of \cite{KWY-2017} with
	\be 
		\Delta_{g_t} v_t-\(|E_t|^2_{g_t}+|B_t|^2_{g_t}+|\text{div}_{g_t}E_t|\)v_t=0
	\ee 
	then condition \eqref{eq-QDEC2} is preserved throughout this slightly modified flow, and the charged positive mass theorem can be applied throughout the flow.
	
	With these minor modifications to \cite{KWY-2017,KWY-exts-2015} one finds that the sign condition on the charge density is in fact not required, provided that one instead imposes that the condition \eqref{eq-QDEC2} holds, and asks that the charge density be compactly supported. In \cite{KWY-exts-2015}, it is also demonstrated that a counterexample to the charged Riemannian Penrose inequality exists if the charge density is not compactly supported; however, it is not clear that such a counterexample can be constructed satisfying \eqref{eq-QDEC2}. In fact, we conjecture that \eqref{eq-chargedpenrose} holds even if the hypothesis of compactly supported charge is dropped, provided that \eqref{eq-QDEC2} holds, which amounts to asking that the mass density of matter is at least equal to the magnitude of the charge density, pointwise. If we instead consider the matter as being discrete particles, each with a mass $m_i$ and charge $q_i$ satisfying $m_i\geq|q_i|$ then the Penrose heuristic argument can still be applied. Specifically, no particles escaping to infinity can increase the quantity $m^2-Q^2$ so we expect that \eqref{eq-chargedpenroseareaform} should hold for any such configuration of matter. Since the condition \eqref{eq-QDEC2} is roughly a continuous version of the condition $m_i\geq|q_i|$ imposed on discrete particles, it seems reasonable to conjecture that \eqref{eq-chargedpenroseareaform} holds under such a hypothesis.

\section*{Acknowledgements}
I would like to thank Marcus Khuri for useful discussions, and in particular for clarifying the modifications to the work in \cite{KWY-2017} that are required to ensure that the charged Riemannian Penrose inequality holds under the stronger energy condition, as discussed above. I would like to thank Armando Cabrera Pacheco for helpful comments on an early draft of this article. I would also like to thank the anonymous reviewers for detailed and helpful reviews that led to improvements in the presentation of this article.

\bigskip

\appendix
\section{Magnetic fields and magnetically charged matter}\label{AMagnetic}
	Throughout this note we have avoided considering magnetic fields for the sake of exposition, in part because usually one imposes that the magnetic field should be divergence free; that is, no magnetically charged matter. However, it is well-known that analogous results hold for magnetically charged black holes and indeed one could even permit magnetically charged matter. In this context, we think of a magnetic field as being a vector field $B$ satisfying $\nabla\cdot B=\rho_B$, where $\rho_B$ is the magnetic charge density. One would usually set $\rho_B$ to zero, however one could imagine permitting a nonzero magnetic charge density and mathematically this is entirely analogous to the electric field. In particular, this means the analysis is the same. One interesting thing to note is that in the presence of both electric an magnetic fields, the electromagnetic field carries linear momentum, so the dominant energy condition in this case would become
	\be \label{eq-QDEC3}
		R(g)-2\(|E|^2+|B|^2\)\geq 4|E\times B|.
	\ee 
	However, since we do not include linear momentum in the charged Riemannian Penrose inequality, we simply ask that
		\be \label{eq-QDEC4}
	R(g)-2\(|E|^2+|B|^2\)\geq 0.
	\ee 
	In what follows, we will now consider a \emph{magnetically charged asymptotically flat manifold} to be the same as a charged asymptotically flat manifold used above, with the addition of a vector field $B$ corresponding to the magnetic field. We define the magnetic charge of a surface $\S$ by
\be 
Q_B(\S)=\frac{1}{4\pi}\int_\S B\cdot n\,dS,
\ee 
	with the charge at infinity defined similarly. The appropriate magnetically charged Hawking mass of a topological sphere $\S$ is
	\be 
	\m_H^{MCH}(\S)=\( \frac{|\S|}{16\pi}\)^{\frac12}\( 1+\frac{4\pi Q(\S)^2+4\pi Q_B(\S)^2}{|\S|}-\frac{1}{16\pi}\int_\S H^2\, dS \).
	\ee 
	
	Following the arguments above, one obtains the following analogues of Theorem \ref{thm-BHQ} and Theorem \ref{thm-inftyQ}, respectively.
	
	\begin{thm}\label{thm-BHQ-A}
		Let $(M,g,E,B)$ be a magnetically charged asymptotically flat 3-manifold satisfying \eqref{eq-QDEC4}, let $\S$ be an area outer minimising sphere and assume there are no closed minimal surfaces in $M$ except possibly $\S$. Assume that exterior to $\S$ it holds that $Q(\S)\nabla\cdot E\geq0$ and $Q_B(\S)\nabla\cdot B\geq0$.
		
		Then
		\be \label{eq-Hawkingineq-A}
		\m_H^{MCH}(\S)\leq\m_{ADM},
		\ee 
		and in particular, if $\S$ is an outermost horizon then
		
		\be\label{eq-chargedpenrose2-A} 
		\m_{ADM}\geq \(\frac{|\S|}{16\pi}\)^{\frac12}\(1+\frac{4\pi Q^2+4\pi Q^2_B}{ |\S|}\),
		\ee 
		where $Q=Q(\S)$ and $Q_B=Q_B(\S)$ are the black hole charges.
		
		Furthermore, equality holds if and only if $\S$ is a round sphere in a Reissner--Nordstr\"om manifold, possibly with both electric and magnetic charges.
	\end{thm}

	\begin{thm} \label{thm-inftyQ-A}
		Let $(M,g,E,B)$ be a magnetically charged asymptotically flat 3-manifold satisfying \eqref{eq-QDEC4}, let $\S$ be an outermost minimal sphere and denote by $Q$ and $Q_B$ the total electric and magnetic charges respectively, computed at infinity. Assume further that exterior to $\S$ it holds that $Q\nabla\cdot E\leq0$, $Q_B\nabla\cdot B\leq0$ and there are no closed minimal surfaces in $M$ except for $\S$.
		
		Then
		\bee
		\m_{ADM}\geq \(\frac{|\S|}{16\pi}\)^{\frac12}\(1+ \frac{4\pi Q^2+4\pi Q_B^2}{ |\S|}\),
		\eee 
		Furthermore, equality holds if and only if $\S$ is a round sphere in a Reissner--Nordstr\"om manifold, possibly with both electric and magnetic charges.
	\end{thm}

\begin{remark}
	In \cite{DK-2015}, the rigidity statement is given in the case of no magnetic field; however, for the same reasons as above, the proof clearly remains valid if magnetic fields are present (cf. \cite{KWY-2017}).
\end{remark}

\end{document}